%% file: paper.tex
\documentclass[conference,compsoc]{IEEEtran}
\ifCLASSOPTIONcompsoc
  \usepackage[nocompress]{cite}
\else
  \usepackage{cite}
\fi
\ifCLASSINFOpdf
  \usepackage[pdftex]{graphicx}
\else
\fi
\usepackage{amsmath,amssymb}

\usepackage{xspace}

\usepackage{placeins}

\usepackage{amsthm}
\usepackage{comment}
\usepackage{enumitem}
\newtheorem{theorem}{Theorem}
\newtheorem{lemma}{Lemma}

\usepackage[
n ,
advantage ,
operators ,
sets ,
adversary ,
landau ,
probability ,
notions ,
logic ,
ff ,
mm,
primitives ,
events ,
complexity ,
oracles ,
asymptotics ,
keys
 ]{ cryptocode }

\usepackage[style=long,nonumberlist,toc,xindy,acronym,nomain]{glossaries}
\glsdisablehyper{}
\input{acronyms}
\makeglossaries

\newcommand{\topMAC}{top \gls{mac}\xspace}

\usepackage{tcolorbox}
\usepackage{algorithm} 
\usepackage{listings}
\usepackage{xcolor}

\definecolor{codegreen}{rgb}{0,0.6,0}
\definecolor{codegray}{rgb}{0.5,0.5,0.5}
\definecolor{codepurple}{rgb}{0.58,0,0.82}
\definecolor{backcolour}{rgb}{0.95,0.95,0.92}

\lstdefinestyle{mystyle}{
    backgroundcolor=\color{backcolour},   
    commentstyle=\color{codegreen},
    keywordstyle=\color{magenta},
    numberstyle=\tiny\color{codegray},
    stringstyle=\color{codepurple},
    basicstyle=\ttfamily\footnotesize,
    breakatwhitespace=false,         
    breaklines=true,                 
    captionpos=b,                    
    keepspaces=true,                 
    numbers=left,                    
    numbersep=5pt,                  
    showspaces=false,                
    showstringspaces=false,
    showtabs=false,                  
    tabsize=2
}

\usepackage{amsthm}
\usepackage{algpseudocode}
\usepackage[separate-uncertainty=true,multi-part-units=single]{siunitx}

\usepackage{hyperref}
\usepackage{cleveref}

\crefname{lemma}{Lemma}{Lemmas}
\newif\ifarchive

\begin{document}
\title{Towards cryptographically-authenticated in-memory data structures}

\author{
    \IEEEauthorblockN{Setareh Ghorshi}
    \IEEEauthorblockA{University of Waterloo \\Canada \\ ghorshi.setareh@gmail.com}
    \and
    \IEEEauthorblockN{Lachlan J.\ Gunn}
    \IEEEauthorblockA{Aalto University \\ Finland\\ lachlan@gunn.ee}
    \and
    \IEEEauthorblockN{Hans Liljestrand}
    \IEEEauthorblockA{University of Waterloo \\ Canada\\ hans@liljestrand.dev}
    \and
    \IEEEauthorblockN{N.\ Asokan}
    \IEEEauthorblockA{University of Waterloo \\Canada \\ asokan@acm.org}
}
\maketitle
\input{sections/00-abstract}
\IEEEpeerreviewmaketitle

\input{sections/10-introduction}
\input{sections/20-background}

\input{sections/30-model}

\input{sections/35-justification}
\input{sections/40-design}

\input{sections/50-implementation}

\input{sections/55-evaluation}
\input{sections/70-related-work}

\input{sections/60-discussion}

\ifCLASSOPTIONcompsoc
\else
\fi
\bibliographystyle{plain}
\bibliography{IEEEabrv,paper}

\input{sections/appendix}

\end{document}

%% file: acronyms.tex
\newacronym{aes}{AES}{Advanced Encryption Standard}
\newacronym{api}{API}{Application Programming Interface}
\newacronym{aslr}{ASLR}{Address Space Layout Randomization}
\newacronym{bti}{BTI}{Branch Target Identification}
\newacronym{dop}{DOP}{Data-Oriented Programming}
\newacronym{cet}{CET}{Control-flow Enforcement Technology}
\newacronym{cfi}{CFI}{Control-Flow Integrity}
\newacronym{cfg}{CFG}{Control-Flow Graph}
\newacronym{cmac}{CMAC}{Cipher-based Message Authentication Code}
\newacronym{cpi}{CPI}{Code-Pointer Integrity}
\newacronym{cve}{CVE}{Common Vulnerabilities and Exposure}
\newacronym{dfi}{DFI}{Data-Flow Integrity}
\newacronym{fifo}{FIFO}{First-In-First-Out}
\newacronym{isa}{ISA}{Instruction-set Architecture}
\newacronym{jop}{JOP}{Jump-oriented programming}
\newacronym{lbc}{LBC}{Light-weight Bounds Checking}
\newacronym{lifo}{LIFO}{Last-In-First-Out}
\newacronym{mac}{MAC}{Message Authentication Code}
\newacronym{os}{OS}{operating system}
\newacronym{pa}{PA}{Pointer Authentication}
\newacronym{pac}{PAC}{Pointer Authentication Code}
\newacronym{ro}{RO}{Random Oracle}
\newacronym{rop}{ROP}{Return Oriented Programming}
\newacronym{stl}{STL}{Standard Template Library}
\newacronym{tls}{TLS}{Transport Layer Security}
\newacronym{wXORx}{$\text{W}\!\oplus\text{X}$}{write-XOR-execute}
\newacronym{xom}{XOM}{Execute-only Memory}

\newcommand{\afast}{\textsf{Adv-Fast}\xspace}
\newcommand{\aslow}{\textsf{Adv-Slow}\xspace}
\newcommand{\asingle}{\textsf{Adv-Single}\xspace}
\newacronym{STL}{STL}{Standard Template Library}

%% file: sections/00-abstract.tex
\begin{abstract}
Modern processors include high-performance cryptographic functionalities such as Intel's AES-NI and ARM's Pointer Authentication that allow programs to efficiently authenticate data held by the program.  Pointer Authentication is already used to protect return addresses in recent Apple devices, but as yet these structures have seen little use for the protection of general program data.

In this paper, we show how cryptographically-authenticated data structures can be used to protect against attacks based on memory corruption, and show how they can be efficiently realized using widely available hardware-assisted
cryptographic mechanisms. We present realizations of secure stacks and queues with minimal overall performance overhead (3.4\%--6.4\% slowdown of the OpenCV core performance tests), and provide proofs of correctness.
\end{abstract}

%% file: sections/10-introduction.tex
\section{Introduction}
Since the time of the Morris Worm~\cite{spafford1989internet}, memory corruption vulnerabilities have been used take control of programs and steal data or cause damage.  This class of vulnerability is still heavily exploited~\cite{miller2019trends}, and much effort has been spent by defenders in hardening programs against memory corruption, and by adversaries in overcoming these defences.

Most of these attacks and defenses focus on programs' control flow: adversaries attempt to overwrite code pointers in the program, such as return addresses or function pointers, in order to make the program deviate from its intended behavior.  Improved defences have reduced adversaries' ability to modify these pointers; one such approach takes advantage of cryptographic functionality \gls{isa} extensions such as ARM's \gls{pa}~\cite{qualcomm2017pointer}.  Apple devices already use \gls{pa} to provide some protection against modification of return addresses on the program stack, and \gls{pa} also can be used to prevent the modification of forward code pointers and data pointers.  However, pointer integrity is not enough, as an adversary can modify a program behavior by overwriting its data, using attacks such as \gls{dop}~\cite{hu2016dataoriented}.  In this paper, we propose the use of cryptographically-authenticated data structures for particularly sensitive program data.

We show how to realize authenticated data structures using widely available hardware primitives like \gls{pa} and AES-NI~\cite{amd-bulldozer-pilediver,gueron2012intel}.
We present an authenticated stack and queue, along with benchmarks to demonstrate their utility in real-world software.  For want of space we omit the design of an authenticated tree.

Our contributions are as follows:
\begin{itemize}
    \item An argument for authenticated data structures (§\ref{sec:justification}).
    \item Designs for an authenticated stack and queue (§\ref{sec:design}).
    \item Implementations of said stack and queue for x86-64 and Arm Aarch64 architectures (§\ref{sec:implementation}).
    \item An evaluation of the authenticated data structures, showing overheads of $3.4\%$ (push+pop) or $6.5\%$ (random access) in realistic software (§\ref{sec:performance-evaluation}), and proving their security (§\ref{sec:security-evaluation}).
\end{itemize}

%% file: sections/20-background.tex
\section{Background}

\subsection{Memory corruption vulnerabilities}
\label{sec:background-wxorx}

Run-time memory errors occur when a memory access reads or writes to unintended memory.
Such errors can allow an attacker to overwrite data or alter program functionality.
Around 70\% of \glspl{cve} are caused by memory safety errors~\cite{miller2019trends}, indicating that they remain a prevalent threat.

Code injection exploits are today widely mitigated on all contemporary \glspl{os} by \gls{wXORx}, such as Windows DEP~\cite{microsoft-dep}.
This has led to the introduction of attacks that alter program behavior without injecting new code by corrupting control data such as return addresses or function pointers: from early return-to-libc attacks~\cite{peslyak1997getting} that can execute arbitrary libc functions, to later generalizations such as  \gls{rop}~\cite{shacham2007geometry} and \gls{jop}~\cite{kornau2009return,bletsch2011jumporiented}.

\subsection{\glsdesc{cfi}}
\label{sec:background-cfi}

\Gls{cfi} ensures that all control-flow transitions within a program conform to a \gls{cfg} generated at compile-time via static analysis\cite{abadi2005controlflow}.
Because the checks are based on static information, they are overly permissive; based on a possibly over approximated points-to set~\cite{abadi2005controlflow}, or in the case of LLVM, based on function type only~\cite{clang14-cfi}.
New features in contemporary processors include hardware-support for coarse-grained \gls{cfi}---e.g., Intel \gls{cet} and ARM \gls{bti}---that simply check that control-flow transfers target a valid function entry point or branch target, but not that the target is the correct one, ensuring that straight-line code cannot be partially executed by jumping into the middle of it.

\subsection{Hardware-assisted cryptography}
\label{sec:background-pa}
\label{sec:background-aesni}

On x86-64 platforms, the AES-NI extension---introduced on Intel CPUs in 2008~\cite{gueron2012intel} and AMD CPUs in 2010~\cite{amd-bulldozer-pilediver}---adds instructions that provide high-performance AES encryption, decryption, and key expansion.

ARMv8.3-A \gls{pa} is an \gls{isa} extension that computes tweakable \glspl{mac}~\cite{qualcomm2017pointer}.
\Gls{pa} includes specialized instructions that compute and embed a \gls{mac} over a pointer directly into unused bits of the pointer itself.
\Gls{pa} also provides a general instruction for generating a \gls{mac} over arbitrary 64-bit values, which can be used to compute \glspl{mac} efficiently over more general data.

%% file: sections/30-model.tex
\section{Problem description}

\subsection{System model}

We consider programs at the level of \emph{basic blocks}, which are sequences of non-branching instructions that terminate with a branch instruction. These chunks of code will always execute linearly, meaning that even an adversary who controls all data used by the code will be unable to influence its control flow within a basic block.

\subsection{Adversary model}\label{sec:adversary-model}
We assume a strong adversary capable of making arbitrary reads and writes to the program's memory space (henceforth referred to only as memory), subject to the memory access control configuration, which we assume to include a \gls{wXORx} policy enforcement (\Cref{sec:background-wxorx}), and the enforcement of coarse-grained \gls{cfi} (\Cref{sec:background-cfi}).
Together, \gls{wXORx} and \gls{cfi} imply that the adversary cannot alter the execution flow of basic blocks, or transfer control to the middle of basic blocks.  Moreover, the targets of register operations are part of the instruction encoding: this means that the attacker cannot overwrite a register by altering the program's data memory, but rather they must find a piece of code that performs the operation that they desire.

The result is that an adversary can modify data stored in memory, but cannot alter program code, and can neither read nor write registers directly.  Threads will read inputs from memory and registers, process them according to each a program-defined sequence of operations, write some values to memory, and then jump to another basic block, the target possibly depending on a result of a computation.  We also assume that the adversary cannot compromise the \gls{os} in order to bypass these protections.

We define a family of adversaries based on their ability to time reads and writes; from strongest to weakest:
\begin{itemize}
    \item \afast can read and write any location in memory at any time in the program's execution with perfect accuracy, allowing them to overwrite values spilled to memory for only a small number of cycles.
    
    \hspace{\parindent} \emph{This models the case where an adversary exploits vulnerabilities in a multi-threaded program, and can interact with other program threads so as to synchronize its exploitation with their reads and writes.}
    
    \vspace{1em}
    
    \item \aslow can read and write any location in memory, but without the timing accuracy to alter values written and re-read within a basic block.
    
    \hspace{\parindent} \emph{This models the case where an adversary exploits vulnerabilities in a multi-threaded program, but cannot do so with tight synchronization with other threads.  This is a conservative approximation, since in practice an adversary will be unable to overwrite specific regions of memory even between some basic blocks.}
    
    \vspace{1em}

    \item \asingle can read and write any location in memory when the program counter is at a vulnerable addresses.
    
    \hspace{\parindent} \emph{This models the case where an adversary exploits a vulnerability in a single-threaded program; they can read and write memory when the program reaches the location of a vulnerability, but at all other times they cannot interfere with program data.}
\end{itemize}

\subsection{Objectives}

\newcommand{\goal}[2]{
    \vspace{\baselineskip}
    \par\noindent
    \newdimen\protocolstepwidth
    \protocolstepwidth=\linewidth
    \addtolength\protocolstepwidth{-2\fboxsep}
    \fbox{\begin{minipage}[c]{\protocolstepwidth}\textbf{#1}: \emph{#2}\end{minipage}}
    \vspace{1.0em}
}

Our goal is to use cryptography to protect program data in memory, with as little modification to the program as possible.  We split this into three main objectives.

\goal{Security}{
    Authenticated data structures will behave according to their functional specification when operated on by their methods, or raise an error.
}

A data structure is described in terms of a set of operations that a program can invoke upon it, such as push and pop.  However, an adversary might modify the state of the data structure directly, without using these operations.  Our security requirement is that any attempt to do so will yield an error when the modified data is read.

\goal{Performance}{
    The use of authenticated data structures will not significantly reduce the overall performance of a program.
}

In order for authenticated data structures to be usable, their overhead must not be so high as to degrade the functionality of a program that incorporates them.

\goal{Compatibility}{
    Authenticated data structures will offer a similar interface to their non-authenticated counterparts.
}

The new functionality must not unduly limit the usage of the data structures: authenticated sdata structures must provide equivalent functionality to data in existing programming languages, so as to act as a drop-in replacement.

\subsection{Cryptographic functionality}
We assume that the \gls{isa} provides some cryptographic functionality to programs.  In particular, we suppose that the processor can compute \glspl{mac} $\mathsf{MAC}(x; k)$ without exposing key material $k$ to the attacker, and without storing the result in memory where it can be overwritten by an attacker.

%% file: sections/35-justification.tex
\section{Securing program data flows}\label{sec:justification}

\subsection{Programs from a protocol perspective}

The computational model described in \Cref{sec:adversary-model} can be viewed as a cryptographic protocol: the program defines operations to be taken, with load and store operations indicating messages from or to the adversary.  The details depend upon the adversary model:
\begin{itemize}
    \item \afast can read and alter the contents of memory at any time.  The program uses a fixed number of registers as safe storage, with loads and stores translated to messages from and to the adversary, respectively.
    \vspace{0.15em}
    
    \item \aslow can read and alter the contents of memory at any time, but if an address is both written and read within the same basic block, then the read will yield the same value that was written.  This allows the program to spill registers to memory when needed, giving practically unlimited working storage. All other loads, and all stores, are translated to messages to and from the adversary, respectively.
    
    \vspace{0.15em}

    \item \asingle can read and alter the contents of memory only at vulnerable points in the program.  Therefore, the program uses a fixed number of registers as working storage, along with a large amount of working storage representing the program's memory; when the program reaches a vulnerable point in the code, all of memory is sent to the adversary, and replaced with new values received from the adversary.
\end{itemize}

By interpreting program execution in this way, we can turn the same analysis machinery that is used to analyze network protocols to the analysis of of data flows between different parts of the program.

\subsection{Communication between basic blocks}\label{sec:communication}

We analyze the program at the level of basic blocks.  This granularity is convenient because execution is linear within a basic block (i.e., execution always proceeds to the next instruction in memory), so we need not analyze control flow within basic blocks, and the initial state of a basic block is invariably the final state of another basic block.

Data flows into basic blocks in three main ways: from outside, through registers, and through memory.

\subsubsection{Communication from outside}
Most programs will process data received from outside.  The integrity of the program's data therefore depends on its ability to ensure the integrity of the data that it reads in.

The main way for a program to obtain outside data is using system calls, transferring control to the \gls{os}, which performs a task and returns control to the program.  The integrity of a system call result returned via register is assured even in the \afast model.  

A common arrangement is for system calls to write the result to a region of memory specified by the program, and return an error code in a register.  The integrity of this data is not guaranteed against \afast, but can be against \aslow, since the data can be read from memory during the same basic block as it was written by the system call.  The same holds for \asingle, if vulnerable code cannot execute between the time of the system call and the time when the data is read.

Even against \afast, it is possible to ensure the integrity of data from outside the program using cryptography.  If an end-to-end authenticated channel is terminated inside the program, or if cryptographically-authenticated data is read from storage, then its integrity can be verified once it is finally loaded into registers.

\subsubsection{Communication via registers}
Once data has been obtained, it must flow throughout the program.  After jumping from one basic block to another, the register state from the previous block will be preserved, thus allowing data to flow between them.  The integrity of this data flow is assured by the adversary's inability to write to registers directly.  Authentication is provided by the \gls{cfi} mechanism: the initial register state of a basic block can be the final state of any basic block capable of jumping to it.  The predecessors of blocks not marked as indirect branch targets by the \gls{cfi} mechanism can be identified using program analysis---since they must terminate with a direct branch pointing to the start of the block---allowing assumptions to be made about the initial register state.  The predecessors of indirect branch targets can be any basic block terminating in an indirect branch instruction.  Finer-grained \gls{cfi} mechanisms such as that by~\cite{liljestrand2019pac} can further limit the sources of indirect branches, allowing more assumptions to be made about basic blocks' initial register states.

Registers are therefore useful for the transfer of data between basic blocks, but their limited number means that bulk data must be transferred in memory.

\subsubsection{Communication via memory}\label{sec:justification-memory}
Unlike communication via registers, basic blocks can transfer large amounts of data via memory; however, said data can be overwritten by \afast or \aslow.

Since the values read from memory can be chosen arbitrarily by the adversary, another approach is needed to ensure its authenticity.  To remedy this, we propose the use of cryptographically-authenticated data structures, which will allow their users to authenticate data from memory.  Restricting the usage of each data structure to particular basic blocks allows readers of memory to be sure of the basic blocks from which the data originated.

\asingle is more limited than \afast and \aslow, as they cannot alter values not yet written to memory when the vulnerable code is executed. This allows data to be safely transferred via memory, as long as the program's control flow does not pass through the vulnerable code between the time that the data is written and the time that it is read.  E.g., functions that appear only deeper the vulnerable code in the call graph, can safely communicate with one another.  This requires that the developer have some knowledge of where vulnerable code is likely to be, and general data transfer between basic blocks requires the same cryptographic methods as for \afast and \aslow.

The authenticated data structures that we propose in \Cref{sec:design} reduce the state of the data structure to a single \emph{state \gls{mac}} that can be stored in a register.  However, register space is still limited, and a whole program's data structures' state \glspl{mac} cannot be kept in registers throughout execution.  To overcome this, a program can use cryptographically secure data structures recursively: one secure data structure has its state \gls{mac} kept in memory, while the others store their state \glspl{mac} in this top-level data structure.  The programmer then ensures that the top-level state \gls{mac}'s register is used only by the program's secure data structure instrumentation, preventing the adversary from overwriting it by directing the program's control flow to a function that uses this register for other computation.  Each data structure must also include a nonce in its cryptographic computation, to keep \glspl{mac} from one data structure from being replayed against a basic block that attempts to access another; this nonce can be stored in memory alongside the data structure in question, as it will be authenticated by the data structure's state \gls{mac}.

%% file: sections/40-design.tex
\section{Authenticated data structures}\label{sec:design}

We saw in \Cref{sec:communication} that registers allow for only limited secure data flows between basic blocks.  In order to transfer bulk data safely, we use cryptographic techniques to reduce bulk data to a single \gls{mac} that fits into registers.

\subsection{Securing the state \glspl{mac}}\label{stackdesign}
In \Cref{sec:justification-memory}, we proposed that one global data structure be used to store the state \glspl{mac} of each structure.  This data structure needs to provide efficient authenticated random access to the list of state \glspl{mac}.

In general, a Merkle tree can provide an authenticated region of memory of size $n$, with random access requiring $\mathcal{O}(\log n)$ computation.   Our design uses a Merkle tree~\cite{merkle1980protocols} over a `safe storage' region of memory to reduce many state \glspl{mac} to a single \gls{mac} to be stored in a reserved register.

A generic strategy to implement any data structure is to implement it as usual atop an authenticated region of memory, reduced to a single \glspl{mac} by a Merkle tree; however, the $\mathcal{O}(\log n)$ overhead means that this approach is asymptotically slower than the same data structure without authentication.  To eliminate this overhead, we design authenticated versions of specific data structures with cryptography tailored to their access pattern.

In the remainder of this section, we will describe optimized designs for authenticated stacks and queues; the detailed algorithmic descriptions are given in \Cref{sec:spec}.  We have also designed and implemented a authenticated red-black-tree~\cite[p.~308]{clrs}, which we omit here for lack of space.
\subsection{Secure Stack}
Stacks store and retrieve data in \gls{lifo} order. Data is inserted into and read from the top of the stack. The basic operations for a stack are as follows:
\begin{center}
\begin{tabular}{rl}
    \textbf{push} & Pushes data to the top of the stack \\
    \textbf{pop} & Pops the top element out of the stack \\
    \textbf{top} & Returns the element at the top of the stack \\
    \textbf{size} & Returns the number of elements in the stack
\end{tabular}
\end{center}
All stack operations require only $\mathcal{O}(1)$ computation.

For each value added to the stack, a data \gls{mac} is calculated from the data value, the size of the stack, the nonce, the \gls{mac} of the next highest value, and the key $k$:
\begin{align*}
    \mathrm{MAC}_i = \mathsf{MAC}(H(\text{data}), \text{nonce}, \text{size}, \mathrm{MAC}_{i-1}; k) .
\end{align*}

This \gls{mac} is stored in memory, along with the data itself and the nonce used to distinguish \glspl{mac} from different stack instances.
Because the topmost \gls{mac} recursively incorporates all other data \glspl{mac} in the structure, it can serve as the state \gls{mac}, and securing it in a register is sufficient to safely verify the integrity of all the other \glspl{mac}.

In a stack, data is always read from the top, so we can verify the integrity of the topmost value in $\mathcal{O}(1)$ time using the topmost \gls{mac}.  By verifying the top-most data \gls{mac} $\mathrm{MAC}_{i}$, we are also assured of the integrity of the \gls{mac} of the next value $\mathrm{MAC}_{i-1}$, as shown in \Cref{alg:stacktop}; this can then replace the topmost \gls{mac} if the value is to be removed from the stack, as shown in \Cref{alg:stackpop}.
To push a new element onto the stack, a new \gls{mac} $\mathrm{MAC}_{i+1}$ replaces the topmost \gls{mac} in safe storage.

This approach requires $\mathcal{O}(1)$ computation per operation, the same as the operations of a normal stack~\cite[p.~233]{clrs}.

\subsection{Secure Queue}\label{queuedesign}
Queues write and read data in \gls{fifo} order, with elements pushed to the back of the queue and popped from the front:
\begin{center}
\begin{tabular}{rl}
    \textbf{enqueue} & Inserts the elements to the back of the queue \\
    \textbf{dequeue} & Removes the element at the front of the queue \\
    \textbf{front} & Returns the element at the front of the queue \\
    \textbf{back} & Returns the element at the back of the queue \\
    \textbf{size} & Returns the number of the elements in the queue
\end{tabular}
\end{center}
These operations require only $\mathcal{O}(1)$ computation~\cite[p.~235]{clrs}.

The queue differs from the stack in that it has two points of access to the stored data: one at the front, and one at the back.  This means that we cannot use a chained \gls{mac} structure efficiently as in \Cref{stackdesign}, as enqueueing or dequeueing an element will require that all of the hashes in one chain be updated, increasing enqueueing and dequeueing times in an $n$-element queue from $\mathcal{O}(1)$ to $\mathcal{O}(n)$.

The queue admits an alternative implementation that maintains its performance characteristics.  The \glspl{mac} associated with each stored value are not ordered by a \gls{mac} chain, but by an index $i$ incorporated into each data \gls{mac}
\[\mathrm{MAC}_i = \mathsf{MAC}(H(\text{data}), \text{nonce}, i; k) ,\]
which is stored in memory along with the associated data.

Then, only the indices of the oldest and newest values in the queue must be secured; they can be reduced to one state \gls{mac} to be kept in safe storage:
\[\text{state} = \mathsf{MAC}(\text{nonce}, \text{back-index}, \text{front-index}; k). \]

When a new value is enqueued, or a stored value dequeued, the respective index is incremented, and the queue state \gls{mac} updated, shown in \Cref{alg:queueenq,,alg:queuedeq}.

Reading from the queue with \emph{enqueue}, \emph{dequeue}, \emph{front}, or \emph{back} takes place by checking the value's \gls{mac} $\mathrm{MAC}_i$, and verifying its relative position by checking the queue's state \gls{mac} as in~\Cref{alg:queuefront,alg:queueback}.
The number of elements in the queue can be determined by subtracting the front- and back-indices from one another as shown in \Cref{alg:queuesize}.

As with the stack, \glspl{mac} from different queue instances are made distinguishable from one another by incorporating the instance-specific nonce into each \gls{mac}.

%% file: sections/50-implementation.tex
\section{Implementation}
\label{sec:implementation}

We developed a proof-of-concept to evaluate the performance and backwards-compatibility of our approach\footnote{Source code will be made available at \url{https://github.com/ssg-research/authenticated-data-structures}.}.

Our authenticated data structures support arbitrary element types, including complex objects.
Our default element hashing algorithm uses software-only SHA-2, and treats objects as flat data-structures.
Complex data-structures is supported by allowing the default hash computation to be replaced by the programmer.

To generate the \glspl{mac} and Merkle tree root computations we use architecture-specific calculations.
On x86-64, we use AES-NI instructions to implement the CMAC \gls{mac} algorithm~\cite{rfc4493}.
On Aarch64, we use \gls{pa}'s \texttt{pacga} instruction to compute a \gls{mac} over the data 64-bits at a time.
In both cases, the \gls{mac} key is generated at program start, and stored in registers \texttt{xmm5-xmm15} on x86 (which we reserve for key storage), and in dedicated \gls{pa} key registers on ARM.

The Merkle tree string state \glspl{mac} is kept in thread-local storage, and its root \gls{mac}  register \texttt{r13} on x86-64 and \texttt{x28} on AArch64.  
To prevent manipulation, these registers are reserved and cannot be used by other code.
Since threads cannot share registers, our proof-of-concept authenticated data structures cannot be shared between threads.

\noindent\textbf{C++ \glsdesc{api}}\label{sec:object-wrappers},
The authenticated stack and queue subclass \texttt{std::stack} and \texttt{std::queue} from the C++ standard library, respectively.
They incorporate a second container of the parent type, which holds the associated \glspl{mac}.
The operations described in \Cref{sec:design} are then implemented in terms of operations on the container and its \gls{mac} storage structure.

We follow the C++ container library \gls{api} to maintain compatibility with existing code.
However, the \gls{api} provides direct access to contained objects by returning a pointer which can then be used to modify its values without using the container \gls{api} without updating the \gls{mac}, causing an integrity check to fail when the modified element is accessed through the \gls{api}.

We address this by modifying the stack \gls{api} to return a smart pointer to elements in the data structure, which allows the programmer to manually trigger a \gls{mac} update after modifications to the element.

%% file: sections/55-evaluation.tex
\section{Evaluation}
\subsection{Performance \& Compatibility}\label{sec:performance-evaluation}

We have tested the performance of our implementation both with microbenchmarks of individual operations, and in a larger code-base by using it to replace data structures in the OpenCV performance tests and sample applications. 

\begin{table}[]
    \centering
    \begin{tabular}{|c|c|c|c|}
        \hline
         Benchmark & Unauthenticated & Authenticated & Overhead  \\
         \hline
         Stack & \SI{11.211 \pm 0.019}{\nano\second} & \SI{16.853 \pm 0.964}{\micro\second} & \SI{1503}{\times} \\
         Queue & \SI{11.128 \pm 0.032}{\nano\second} & \SI{16.793 \pm 0.753}{\micro\second} & \SI{1509}{\times} \\
         \hline
    \end{tabular}
    \vspace{0.25em}
    \caption{Mean execution times for authenticated data structure operations and 95\% confidence intervals. Each benchmark (500 insertions, 500 removals) was run $10^4$ times.}
    \label{tbl:benchmarks}
\end{table}

In order to test the overhead of a single operation, we measured the overhead of integer insertion operations in the original and authenticated data structures, shown in \Cref{tbl:benchmarks}.

Despite this overhead for individual operations, real programs do not suffer unacceptable slowdown.  We built OpenCV with \texttt{gcc} 11.1.1, modified to use our secure stack and queue implementations by default.  OpenCV provides performance tests that can be used to test the performance impact of our secure data structures. We use the C++ samples provided by OpenCV source code to measure the performance overhead in a more representative application.  We were able to use our implementation of the stack and queue as drop-in replacements for the C++ standard library implementations of the methods described in \Cref{sec:design}; however, as we are able to do so only for structures used within a single thread, and some other methods require the use of the smart pointers from \Cref{sec:object-wrappers}, we conclude that the compatibility requirement is partially met.

We ran three times all OpenCV core performance tests that did not use functionality containing inline assembly overwriting the reserved registers, using authenticated and unauthenticated data structures. For each test run, we compute geometric means of the ratios of the running times for the authenticated and unauthenticated data structures, yielding a \textbf{3.42\% overhead} for the methods in \Cref{sec:design}, and \textbf{6.42\%} with tests modified to use smart pointers for authenticated random access.

\subsection{Security}\label{sec:security-evaluation}

Our security requirement states that the authenticated data structures must behave according to their specification, or raise an error.  In practice, this means that once push or enqueue operations are invoked to insert a value, then pop or dequeue operations must yield the same values.

We briefly sketch arguments for the security of each data structure; detailed proofs appear in \Cref{sec:security-proofs}.

The security of the stack derives from the fact that the stack's state \gls{mac} indirectly incorporates every value in the chain.  Verification of the highest value's \gls{mac} authenticates both the most topmost value, as well as the \gls{mac} of the second-topmost value, due to the presumed collision-resistance of the \gls{mac}.  The security of the remaining values is authenticated recursively by the same argument.

The queue's state \gls{mac} authenticates the indices of the oldest and newest values in the queue.  Thus, enqueue and dequeue operations will always use sequential front-indices and back-indices when adding and removing values from the queue, respectively.  The enqueue operation will therefore produce only a single data \gls{mac} with any particular combination of nonce and index, and by the collision-resistance of the \gls{mac}, this will authenticate the enqueued value.  Since only one valid \gls{mac} with a given index and nonce will ever be produced, the dequeue operation can be certain that a valid \gls{mac} authenticates the correct value for this index, and therefore will always yield the correct value.

We therefore conclude that our designs meet the security requirement:  authenticated stacks and queues will either behave as a stack or queue from the perspective of software that invokes their interfaces, or they will raise an error.

The security of our implementations depend on their ability to securely execute the algorithms specified in our design despite the interference of the adversary.  If we assume our implementation and any caller-overridden element-hashing function to not contain any vulnerabilities, then \asingle cannot interfere with its execution of the data structure algorithms.  Limitations of the C++ programming language mean that we cannot guarantee the security of the implementation against \afast or \aslow, as we do not know how it will be split across basic blocks, as discussed further in \Cref{sec:compiler-capabilities}.

%% file: sections/70-related-work.tex
\section{Related work}

Cryptographically-authenticated data structures have already seen widespread use on the internet, especially during the 2010s, when web cryptography became pervasive: lists of X.509 certificates~\cite{x509}, linked by signatures, are combined with Merkle-tree-based Certificate Transparency logs~\cite{rfc6962} to protect the vast majority of web traffic today.

While traditional \gls{cfi} approaches (\Cref{sec:background-cfi} work by verifying that control-flow transitions are expected; a recent approach is to instead protect the integrity of control-data.
This can be done using isolation, such as shadow stacks~\cite{chiueh2001rad} or a safe stack~\cite{kuznetsov2014codepointer}, but later approaches such as CCFI~\cite{mashtizadeh2015ccfi} employ cryptography to ensure integrity of control-data.
ARM \gls{pa} (\Cref{sec:background-pa}) provides hardware-support for CCFI-like pointer integrity using \glspl{mac}, and has been shown to achieve strong security return addresses by storing return addresses in a cryptographically verifiable stack similar to our secure stack implementation~\cite{liljestrand2021pacstack}.

%% file: sections/60-discussion.tex
\section{Discussion}\label{sec:discussion}

\subsection{Multithreading}
In \Cref{sec:justification} we discussed how different parts of the program can safely communicate via registers as a thread jumps between basic blocks.  However, this approach cannot be used for communication between threads, since they do not share registers.  This means that our implementations from \Cref{sec:implementation} are appropriate only when the data structures are used by only a single thread.  Reads by other threads will fail, as their registers will not contain the correct state \gls{mac}.

A multi-threaded data structure implementation needs to replicate the hash between all of the threads that access it.  This may be achieved in several ways: consensus protocols provide safe replication, but their large numbers of synchronization points will be slow.  \Gls{os} support may provide better performance, at the cost of backwards compatibility.

\subsection{Compiler capabilities}\label{sec:compiler-capabilities}
The design in \Cref{sec:design} can secure data against all of the adversaries from \Cref{sec:adversary-model}, but whether a program is secure against \afast, \aslow, or \asingle depends on its implementation.  However, it is challenging to write a program in a high-level language such that the compiler will generate secure code that is secure against any of these adversaries: compilers will spill registers to memory whenever needed, with the programmer having no way to prevent this.  Functions that are able to accept large data structures may force the caller to write them into memory, even when they are small enough to be passed in a register.  This is a problem even within our implementation, as hash functions must accept data of arbitrary size.  Programmers may indicate that functions are to be inlined, meaning that calls to secure data structure methods can be placed into the same basic block as the calling code, allowing arguments to be passed securely despite \aslow.

To overcome these issues will require compiler and possibly language support.  Important features are the ability to reliably pin values to register storage, perhaps allowing automatic spillage to memory using an authenticated data structure, and a secure calling convention that will allow larger quantities of data to be streamed via registers or passed via memory with cryptographic protection.  Such advances will allow programmers to more confidently write code that will be secure against the more powerful \afast and \aslow.

%% file: sections/appendix.tex
\appendix

\subsection{Authenticated stack specification}\label{sec:spec}

\begin{algorithm}
	
	\caption{authenticated-stack.top()} \label{alg:stacktop}
	\begin{algorithmic}[1]
	\item x = data-stack.top()
    \item \text{state-mac} = get-state-mac()
    \If{state-mac == $\mac_k$(hash(x), nonce, size, mac-stack.top())}
        \Return x
    \Else
        \State exception("MAC authentication error.")
    \EndIf
	\end{algorithmic} 
\end{algorithm}

\begin{algorithm}
	
	\caption{authenticated-stack.pop()} \label{alg:stackpop}
	\begin{algorithmic}[1]
	\item x = data-stack.top()
    \item \text{state-mac} = get-state-mac()
    \If{state-mac == $\mac_k$(hash(x), nonce, size, mac-stack.top())}
        \State data-stack.pop()
        \State insert(mac-stack.top())
        \State mac-stack.pop()
    \Else
        \State exception("MAC authentication error.")
    \EndIf
	\end{algorithmic} 
\end{algorithm}

\begin{algorithm}
	\caption{authenticated-stack.push(x)} \label{alg:stackpush}
	\begin{algorithmic}[1]
	 \item \text{state-mac} = get-state-mac()
    \item data-stack.push(x)
    \item size = size + 1
    \item \text{mac-stack}.push(\text{state-mac})
    \item \text{state-mac} = $\mac_k$(hash(x), nonce, size, state-mac)
    \item \text{insert(state-mac)}
	\end{algorithmic} 
\end{algorithm}

\begin{algorithm}
	
	\caption{authenticated-stack.size()} \label{alg:stacksize}
	\begin{algorithmic}[1]
	\If{size $\neq$ 0}
	\State top()
	\Return size
	\Else
	  \If{nonce == get-state-mac()}
	  \\
        \Return size
    \Else
        \State exception("MAC authentication error.")
    \EndIf
	\EndIf
	\end{algorithmic} 
\end{algorithm}

\FloatBarrier
\subsection{Authenticated queue specification}

\begin{algorithm}
	\caption{authenticated-queue.enqueue(x)} \label{alg:queueenq}
	\begin{algorithmic}[1]
    \If{$\mac_k$(nonce, front-index, back-index) $\neq$ get-state-mac()}
        \State exception("MAC authentication error.")
    \EndIf
    \item data-queue.enqueue(x)
    \item back-index = back-index + 1
    \item \text{mac-queue}.enqueue($\mac_k$(hash(x), nonce, back-index))
    \item \text{insert($\mac_k$(nonce, front-index, back-index))}
	\end{algorithmic} 
\end{algorithm}
	\begin{algorithm}
	\caption{authenticated-queue.dequeue()} \label{alg:queuedeq}
	\begin{algorithmic}[1]
	\If{$\mac_k$(nonce, front-index, back-index) $\neq$ get-state-mac()}
        \State exception("MAC authentication error.")
    \EndIf
	\item data-queue.dequeue()
	\item mac-queue.dequeue()
	\item front-index = front-index + 1
	\item \text{insert($\mac_k$(nonce, front-index, back-index))}
	\end{algorithmic} 
	\end{algorithm}
	
	\begin{algorithm}
	\caption{authenticated-queue.front()} \label{alg:queuefront}
	\begin{algorithmic}[1]
	\If{$\mac_k$(nonce, front-index, back-index) $\neq$ get-state-mac()}
        \State exception("MAC authentication error.")
    \EndIf
	\item x = data-queue.front()
    \If{mac-queue.front() == $\mac_k$(hash(x), nonce, front-index)}
        \Return x
    \Else
        \State exception("MAC authentication error.")
    \EndIf
	\end{algorithmic} 
\end{algorithm}
	
\begin{algorithm}
	\caption{authenticated-queue.back()} \label{alg:queueback}
	\begin{algorithmic}[1]
	\If{$\mac_k$(nonce, front-index, back-index) $\neq$ get-state-mac()}
        \State exception("MAC authentication error.")
    \EndIf
	\item x = data-queue.back()
    \If{mac-queue.back() == $\mac_k$(hash(x), nonce, back-index)}
        \Return x
    \Else
        \State exception("MAC authentication error.")
    \EndIf
	\end{algorithmic} 
\end{algorithm}

\begin{algorithm}
	\caption{authenticated-queue.size()} \label{alg:queuesize}
	\begin{algorithmic}[1]
	\If{$\mac_k$(nonce, front-index, back-index) $=$ get-state-mac()}
	\\
        \Return back-index - front-index + 1
    \Else
        \State exception("MAC authentication error.")
    \EndIf
	\end{algorithmic} 
\end{algorithm}
\FloatBarrier

\subsection{Proofs of security}\label{sec:security-proofs}

We first introduce a security game $\textrm{MAC-Collision}^\adv_{\mac_k}$, which is used to define the collision-resistance property of the process-specific \gls{mac} $\mac_k$. 
 
 \begin{tcolorbox}
\begin{pcimage}
\procedureblock[linenumbering,space=auto, mode=pseudocode  ]{$\textrm{MAC-Collision}^\adv_{\mac_k}(q)$}{
	\pcfor i \in {1,...,q}\\
    (x, y) \leftarrow \adv.choose()\\
    \adv.receive(\mac_k(x,y))\\
	\pcendfor\\
	(x'', x', y) \leftarrow A.\text{gen-collision}()\\
	\pcif x' \neq x'' \wedge \mac_k(x',y) = \mac_k(x'',y)\\
	\pcreturn 1\\
	\pcendif\\
	\pcreturn 0
}
 \end{pcimage}
\end{tcolorbox}

 The \gls{mac} used in our implementation from \Cref{sec:implementation} is computed using \gls{pa}'s \texttt{pacga} instruction on the Aarch64 platform, and an AES-NI-based \gls{cmac} implementation on the x86 platform. Assuming that both functions are pseudo-random with respect to their keys, the most efficient way for the adversary to find a collision is through brute force. According to \cite{menezes2018handbook}, the probability of the adversary finding a collision in a $b$-bit \gls{mac} after collecting $q$ \glspl{mac} is
 \[ Pr[\textrm{MAC-Collision}^\adv_{\mac_k}(q) = 1] = 1 - \frac{2^b!}{(2^b - q)!2^{q.b}} , \]
  and on average, the adversary would find a collision after collecting
 \[ q = \sqrt{\frac{\pi 2^b}{2}} \; \text{\glspl{mac}}.\]
 
 Therefore, if the length of the \glspl{mac} is long enough, we can assume that to find a collision is small. The value of $b$ is 32~bits for Arm \gls{pa} and 128~bits for AES-NI.

We now proceed to show that \adv's goal of successfully corrupting an authenticated data structure can be reduced to finding a collision.

\subsubsection{Stack}
We define a series of games to prove the security of the authenticated stack. These games demonstrate a scenario in which \adv\xspace uses their control over memory to attempt to change the values in the stack to those with a different hash, without being detected.

$\mathbf{\textbf{Stack-Game-MAC}^\adv_1.}$ This is an attack game against the integrity of the authenticated stack.  The game consists of two parts:
In the first loop, \adv\xspace chooses values to be pushed to the stack, and receives the corresponding state \gls{mac}. \adv\xspace can also empty the stack as needed, to start again from an empty stack.

In the second loop, \adv~carries out their attack, attempting to replace at least one of the elements in the stack with one of a different hash, such that the \glspl{mac} still verify successfully. Otherwise \adv\xspace loses, as either they have failed to change any values to those of a different hash, or they have been detected, causing the program to crash.

\begin{tcolorbox}
 \begin{gameproof}[name=\text{Stack-Game-MAC}^\adv, arg=(q)]
\begin{pcvstack}[ center]
\gameprocedure[ linenumbering,space=auto, mode=pseudocode ,  valign ]{
\label{my:line:stack1}
\pccomment{The { \text{first}} steps represent authenticated stack initialization.}\\
\pccomment{The data- and MAC- stacks are unprotected stacks.}\\
\text{macs-stack} \leftarrow []\\
\text{data-stack} \leftarrow []\\
\text{nonce} \sample \text{random}\\
\text{mac-in-register} \leftarrow \text{nonce}\\
\pccomment{pushed-values is a list that stores the values pushed into }\\
\pccomment{the stack to be later used in the attack loop.}\\
\text{pushed-values} = []
}
\end{pcvstack}
\end{gameproof}
\end{tcolorbox}
\begin{tcolorbox}

 \begin{gameproof}[name= \text{Stack-Game-MAC}^\adv, arg=(q), nr=0]

\begin{pcvstack}[ center ]
\gameprocedure[lnstart=9, linenumbering,space=auto, mode=pseudocode , valign ] {
\pccomment{In the following loop, \adv\xspace experiments with the }\\
\pccomment{stack by pushing $n$ values and} \\
\pccomment{receiving the corresponding top \gls{mac}. After each round of $n$ }\\
\pccomment{push operations, the stack will be emptied to}\\
\pccomment{allow the same process again. The mac validation steps have}\\
\pccomment{been omitted since \adv's }\\
\pccomment{goal is not to attack at this stage.}\\
\pcfor i \in {1,...,q} \pcdo \\
n \leftarrow \adv.\text{stack-choose-n}()\\
\pcfor j \in {1,...,n} \pcdo \\
x \leftarrow \adv.\text{stack-choose}() \\
\pccomment{The next steps represent a push operation.}\\
\text{data-stack}.\text{push}(x)\\
\text{size} \leftarrow \text{size} + 1\\ 
\text{mac-stack}.\text{push}(\text{mac-in-register})\\
\text{mac-in-register} \leftarrow \\ \pcind\mac_k(x, \text{nonce}, \text{size}, \text{mac-in-register})\\
\pccomment{x is inserted {\normalfont in the pushed-values list to keep track of}}\\
\pccomment{ the pushed values.}\\
\text{pushed-values}.\text{insert}(x)\\
\pcendfor \\
\pccomment{\adv\xspace receives the state \gls{mac} and then the}\\
\pccomment{{\normalfont stack is emptied back to its initial state.}}\\
\adv.\text{stack-receive}(\text{mac-in-register})\\
\text{macs-stack} \leftarrow []\\
\text{data-stack} \leftarrow []\\
\text{mac-in-register} \leftarrow \text{nonce}\\
\text{pushed-values} = []\\
\pcendfor \\
\pccomment{{\normalfont In this loop, \adv\xspace attempts to violate the}}\\
\pccomment{integrity of the stack by replacing data and its \gls{mac}.}\\
\pccomment{{\normalfont If the returned data is different from what was originally}}\\
\pccomment{{\normalfont pushed, and the \glspl{mac} verify, \adv}}\\
\pccomment{{\normalfont wins the game; otherwise, they lose.}}\\
\pcforeach x' \in \text{pushed-values}\\
(x'', y) \leftarrow \adv.\text{stack-attack}()\\
\pcif x' \neq x''\\
\pcif \mac_k(x'', y) = \text{mac-in-register}\\
\pcreturn 1\\
\pcelse\\
\pcreturn 0\\
\pcendif\\
\pcendif\\
\pccomment{{\normalfont The next line updates the state and has no}}\\
\pccomment{ effect on the probability of \adv\xspace winning the game.}\\
\text{mac-in-register} \leftarrow \text{macs-stack}.\text{pop}()\\
\pcendif\\
\pcreturn 0
}
 \end{pcvstack}
 \end{gameproof}
\end{tcolorbox}

$\mathbf{\textbf{Stack-Game-MAC}^{\bdv^{\adv}}_2.}$
We introduce a second game, $\textrm{Stack-Game-MAC}^{\bdv^\adv}_2$, which can be reduced to the ${\text{MAC-Collision}^\adv_{\mac_k}(q)}$ game. For this purpose, we also introduce a new adversary, $\bdv^{\adv}$, and replace $\adv$ in $\text{Stack-Game-MAC}_1$ with $\bdv^{\adv}$.

\begin{tcolorbox}

\begin{gameproof}[name= \text{Stack-Game-MAC}^{\bdv^{\adv}}, arg=(q), nr=1]

\begin{pcvstack}[ center]
\gameprocedure[linenumbering,space=auto, mode=pseudocode ,  valign ] {
\label{my:line:stack-mac-attack2-begin}
\bdv^{\adv}.\text{stack-init}()\\
\text{pushed-values} = []\\
\pccomment{Here, \adv\xspace \space is replaced with $\bdv^{\adv}$ which performs the update }\\
\pccomment{steps {\normalfont for} the data structure but cannot} \\
\pccomment{calculate the \glspl{mac}.}\\
\pcfor i \in {1,...,q} \pcdo \\
n \leftarrow \bdv^{\adv}.\text{stack-choose-n}()\\
\pcfor j \in {1,...,n} \pcdo \\
(x, y) \leftarrow \bdv^{\adv}.\text{stack-choose}()\\
\text{mac-in-register} \leftarrow \mac_k(x, y)\\
\text{pushed-values}.\text{push}(x)\\
\pcendfor \\
\pccomment{ in the next step, $\bdv^{\adv}$ receives the top MAC and}\\
\pccomment{ resets the stack. }\\
\bdv^{\adv}.\text{stack-receive}(\text{mac-in-register})\\
\pcendfor\\
\pccomment{Again, \adv\xspace \space is replaced with $\bdv^{\adv}$ who performs the attack}\\
\pccomment{ and state-updating steps.}\\
(x'', y, mac) \leftarrow \bdv^{\adv}.\text{stack-attack}()\\
\label{my:line:stack-mac-attack2}
\pcif \mac_k(x'', y) = mac\\
\pcreturn 1\\
\pcelse\\
\pcreturn 0\\
\pcendif\\
\pcreturn 0
}
 \end{pcvstack}
 \end{gameproof}
\end{tcolorbox}

\begin{tcolorbox}
 \label{B-A}
 \begin{pcvstack}
 \procedure[ mode=pseudocode  ]{\textbf{$\bdv^{\adv}$.\text{stack-init}()} } {
    \text{macs-stack} \leftarrow []\\
    \text{data-stack} \leftarrow []\\
    \text{nonce} \sample \text{random}\\
    \text{mac-in-register} \leftarrow \text{nonce}
     }
     
 \end{pcvstack}
 \end{tcolorbox}

\begin{tcolorbox}
 \begin{pcvstack}[center]
    \procedure[ space=auto, mode=pseudocode   ]{\textbf{$\bdv^{\adv}$.\text{stack-choose}()} } {
    x \leftarrow \adv.\text{stack-choose}()\\
    \text{data-stack}.\text{push}(x)\\
    \text{size} \leftarrow \text{size} + 1\\ 
    \text{mac-stack}.\text{push}(\text{mac-in-register})\\
    \pcreturn (x, (\text{nonce}, \text{size}, \text{mac-in-register}))
     }
     
     \procedure[ space=auto, mode=pseudocode  ]{\textbf{$\bdv^{\adv}$.\text{stack-choose-n}()} } {
    \adv.\text{stack-choose-n}()
     }
     
    \procedure[ space=auto, mode=pseudocode  ]{\textbf{$\bdv^{\adv}$.\text{stack-receive}(mac)} } {
    \adv.\text{stack-receive}(mac)\\
    \text{macs-stack} \leftarrow []\\
    \text{data-stack} \leftarrow []\\
    \text{mac-in-register} \leftarrow \text{nonce}\\
    \text{pushed-values} = []
     }

    \procedure[ space=auto, mode=pseudocode  ,  valign]{\textbf{$\bdv^{\adv}$.\text{stack-attack}()} } {
    \textbf{\pcforeach} \space x' \in \text{pushed-values}\\
    (x'',y) \leftarrow \adv.\text{stack-attack}()\\
    \pcif x' \neq x''\\
    \pcreturn (x'',y,\text{mac-in-register})\\
    \pcendif\\
    \pccomment{State update.}\\
    \pcif \text{mac-stack}.\text{size}() > 0\\
    \text{mac-in-register} \leftarrow \text{mac-stack}.\text{pop}()\\ 
    \pcendfor\\
    \pccomment{If the \adv\xspace does not attempt to attack, at the end, {\normalfont the initial} value}\\
\pccomment{ that was pushed {\normalfont in} the stack is returned with}\\
\pccomment{ the correct values which will verify and the \adv\xspace loses the game. }\\
\pccomment{(Nonce is the mac used for {\normalfont the initial} data mac to}\\ \pccomment{ be calculated so the previous \gls{mac} here is nonce)}\\
    \pcreturn(x',(\text{nonce}, \text{size}, \text{nonce}),\text{mac-in-register})
     }
 \end{pcvstack}
 \end{tcolorbox}

\begin{lemma}\label{claim1stack} 
\begin{align*}
&Pr[\textrm{Stack-Game-MAC}^\adv_1(q) = 1] \\ &\leq Pr[\textrm{Stack-Game-MAC}^{\bdv^{\adv}}_2(q) = 1].    
\end{align*}

\end{lemma}
\begin{proof}
The transition from the first game to the second game involves replacing $\adv$ with $\bdv^{\adv}$, which wraps $\adv$, but adds additional functionality required for the authenticated stack. For instance, steps such as initializing or updating the state of the data structure are performed by $\bdv^{\adv}$. These steps are required for the correct functionality of the data structure but have no effect on the probability of the \adv\xspace winning the game. Accordingly, $\adv$ winning the first game implies that $\bdv^{\adv}$ can also win the second game (since $\bdv^{\adv}$ uses $\adv$ in order to perform the attack) with equal probability.
\end{proof}

\begin{lemma}\label{game2-lemma-stack}
\begin{align*}
    &Pr[\textrm{Stack-Game-MAC}^{\bdv^{\adv}}_2(q) = 1] \\  &\leq Pr[{\textrm{MAC-Collision}^\adv_{\mac_k}}(q) = 1]
\end{align*}
\end{lemma}
\begin{proof}
We can reduce $\textrm{Stack-Game-MAC}^{\bdv^{\adv}}_2$(q) to ${\textrm{MAC-Collision}^\adv_{\mac_k}}(q)$. From lines 19--20 of
$\textrm{Stack-Game-MAC}^{\bdv^\adv}_2$, winning $\textrm{Stack-Game-MAC}^{\bdv^\adv}_2$(q) requires that  $\bdv^{\adv}$ find a collision in the state \glspl{mac} so that the authentication passes successfully when comparing the \gls{mac} of the proposed top element by \adv\xspace with the \gls{mac} in the register. Moreover, since $\bdv^{\adv}$ winning $\textrm{Stack-Game-MAC}^{\bdv^\adv}_2(q)$ implies that $\adv$ can win  ${\textrm{MAC-Collision}^\adv_{\mac_k}}(q)$ with the same probability (we can replace $\adv$ with $\bdv^{\adv}$ in the collision game and win the game), we obtain the bound above. 
\end{proof}

\begin{theorem}[Stack Security]
\label{stack-theorem}
An adversary capable of overwriting the values stored in memory memory can
violate the integrity of the authenticated stack with a probability of at most:
\begin{equation}
  \begin{tabular}{l}
                \ensuremath{\textrm{Stack-Game-MAC}^\adv_1(q)}
                \end{tabular}
        \leq Pr[{MAC-Collision}^\adv_\mac(q)]
\end{equation}

\end{theorem}
\begin{proof}
We can replace $\textrm{Stack-Game-MAC}^\adv_1(q)$ with $\textrm{Stack-Game-MAC}^{\bdv^\adv}_2(q)$ using the reformulation from \Cref{claim1stack}. We then apply \Cref{game2-lemma-stack}, yielding the bound above.
 \end{proof}

\subsubsection{Queue}

We create two sets of games to prove the security of the authenticated queue, this time in the random oracle model. In the first games, \textrm{Queue-Game-Index-MAC}, the adversary attempts to change the indices in order to substitute one entry with another. The second set of games, \textrm{Queue-Game-Data-MAC}, deal with the possibility that the adversary tries to violate the integrity of the authenticated queue by replacing an element and its \gls{mac} with a value having a different hash, without changing the indices. Similarly to the stack, proving the inability of the adversary to modify the content of the authenticated queue shows that the security requirement is satisfied for this data structure.

$\mathbf{\textbf{\textrm{Queue-Game-Index-MAC}}^\adv_1.}$  Similarly to the authenticated stack, we assume an adversary $\adv$ who has arbitrary read/write access to data stored in memory. Accordingly, we define the game such that $\adv$ chooses values to be enqueued in or dequeued from the authenticated queue and then receives the corresponding \glspl{mac}. After $q$ rounds of performing this process, the adversary tries to attack the authenticated queue by replacing at least one of the elements with a different value without being noticed.

\begin{tcolorbox}
\label{queue-index-game-1}
 \begin{gameproof}[name= \text{Queue-Game-Index-MAC}^\adv, arg=(q)]

\begin{pcvstack}[  center ]
\gameprocedure[ linenumbering, space=auto, mode=pseudocode ,  valign ] {
\pccomment{The following steps represent authenticated-queue initialization.}\\ \pccomment{ Data and \gls{mac} queues are unprotected.}\\
\text{macs-queue} \leftarrow []\\
\text{data-queue} \leftarrow []\\
\text{nonce} \sample \text{random}\\
\text{back-index} \leftarrow 0\\
\text{front-index} \leftarrow 1\\
\text{mac-in-register} \leftarrow \mac_k(\text{nonce}, \text{back-index}, \\ \text{front-index})\\
\pccomment{enqueued-values is a list that stores the values enqueued into }\\ \pccomment{authenticated queue to be used {\normalfont in} the attack loop.}\\
\text{enqueued-values} = []\\
\pccomment{In the following loop, the adversary experiments with the}\\ \pccomment{ authenticated queue through enqueue and dequeue} \\
\pccomment{operations {\normalfont in} order to collect corresponding \glspl{mac} {\normalfont for}}\\ \pccomment{ different index values. The \gls{mac} validation steps}\\
\pccomment{have been omitted since the adversary's goal is not to attack }\\ \pccomment{at this stage.}\\
\pcfor i \in {1,...,q} \pcdo \\
(x, enqueue) \leftarrow \adv.\text{queue-choose-index-attack}()\\
\pcif enqueue \\
\pccomment{The following steps represent \text{authenticated-queue}.\text{enqueue}.}\\
\text{data-queue}.\text{enqueue}(x)\\
\text{back-index} \leftarrow \text{back-index} + 1\\ 
\text{mac-queue}.\text{enqueue}(\mac_k(x,(\text{nonce}, \\ \pcind\text{back-index})))\\
\text{mac-in-register} \leftarrow \\ \pcind \mac_k(\text{nonce}, \text{back-index}, \text{front-index})\\
\text{enqueued-values}.\text{enqueue}(x)\\
\adv.\text{queue-receive}(\text{mac-queue}.\text{back}(),\\ \pcind\text{mac-in-register})\\
\pcelse\\
\pccomment{The following steps represent \text{authenticated-queue}.\text{dequeue}.}\\
\text{data-queue}.\text{dequeue}()\\
\text{mac-queue}.\text{dequeue}()\\
\text{front-index} \leftarrow \text{front-index} + 1\\ 
\text{mac-in-register} \leftarrow \mac_k(\text{nonce}, \text{back-index},\\ \pcind\text{front-index})\\
\text{enqueued-values}.\text{dequeue}()\\
\adv.\text{queue-receive}(\text{mac-queueu}.\text{front}(),\\ \pcind\text{mac-in-register})\\
\pcendif \\
\pcendfor\\
\pccomment{In the following loop, \adv attempts to violate the}\\ \pccomment{ integrity by replacing data, its \gls{mac}, and the}
}
 \end{pcvstack}
 \end{gameproof}
\end{tcolorbox}
\begin{tcolorbox}

 \begin{gameproof}[name= \text{Queue-Game-Index-MAC}^\adv, arg=(q), nr=0]

\begin{pcvstack}[ center ]
\gameprocedure[lnstart=45, space=auto, linenumbering, mode=pseudocode ,  valign ] {
\pccomment{indices. If the returned data is different from what was }\\ \pccomment{originally enqueued, and the \glspl{mac} verify, the}\\
\pccomment{\adv wins the game, otherwise, they lose.}\\
\pcforeach x' \in \text{enqueued-values}\\
(x'', mac, k_1, k_2) \leftarrow \adv.\text{queue-index-attack}()\\
\pcif x' \neq x''\\
\pcif mac = \mac_k(x'',(\text{nonce}, k_1)) \wedge \\ \mac_k(\text{nonce}, k_2, k_1) = \text{mac-in-register}\\
\pcreturn 1\\
\pcelse\\
\pcreturn 0\\
\pcendif\\
\pcendif\\
\text{front-index} \leftarrow \text{front-index} + 1\\ 
\text{mac-in-register} \leftarrow \mac_k(\text{nonce}, \text{back-index}, \\ \text{front-index})\\
\pcendfor\\
\pcreturn 0
}
 \end{pcvstack}
 \end{gameproof}
\end{tcolorbox}

$\mathbf{\textbf{\textrm{Queue-Game-Index-MAC}}^\adv_2.}$
Similarly to the approach used for the stack, we define a second game
which can be reduced to the ${\textrm{MAC-Collision}^\adv_\mac}(q)$ game. For this purpose, we define a new adversary, $\bdv^{\adv}$,
and replace $\adv$ in $\textrm{Queue-Game-Index-MAC}^\adv_1$ with $\bdv^{\adv}$.

\begin{tcolorbox}
\label{queue-game-2-index}
 \begin{gameproof}[name=\text{Queue-Game-Index-MAC}^{\bdv^{\adv}}, nr=1, arg=(q)]
\label{queue-game-2-index}
\begin{pcvstack}[ center ]
\gameprocedure[linenumbering, space=auto, mode=pseudocode ] {
\bdv^{\adv}.\text{queue-init}()\\
\pccomment{In this loop, we replace the enqueue and dequeue steps with the}\\ \pccomment{ $\bdv^{\adv}$.\text{queue-choose}()}\\ \pccomment{function which performs the same steps.}\\
\pcfor i \in {1,...,q} \pcdo \\
(x, y, y') \leftarrow \bdv^{\adv}.\text{queue-choose}()\\
\bdv^{\adv}.\text{queue-receive}(\mac_k(x,y),\mac_k(y'))\\
\pcendfor\\
\pccomment{In this part, similar to the previous steps \adv\xspace \space}\\ \pccomment{ is replaced with $\bdv^{\adv}$ \space who also performs the}\\ \pccomment{authenticated-queue related steps and the loop. The step}\\ \pccomment{to update the mac-in-register is removed since it's just}\\ \pccomment{a state update and doesn't affect the probability.}\\
(x'', mac_1,mac_2, y, y') \leftarrow \bdv^{\adv}.\text{queue-attack}()\\
\pcif mac_1 = \mac_k(x'',y) \wedge \mac_k(y') = mac_2\\
\pcreturn 1\\
\pcelse\\
\pcreturn 0\\
\pcendif\\
\pcreturn 0
 }
 \end{pcvstack}
 \end{gameproof}
 \end{tcolorbox}

\begin{lemma}\label{claim1-queue}
\begin{align*}
&Pr[\textrm{Queue-Game-Index-MAC}^\adv_1(q) = 1] \\ &\leq Pr[\textrm{Queue-Game-Index-MAC}^{\bdv^{\adv}}_2(q) = 1].
\end{align*}
\end{lemma}
\begin{proof}
The transition from the first game to the second game involves replacing $\adv$ with $\bdv^{\adv}$, which wraps $\adv$ but also performs the computations required for queue operations such as initializing and updating the state.  Accordingly, we obtain the given bound since $\adv$ winning the first game with some probability implies that $\bdv^{\adv}$ can win the second game with the same probability.
\end{proof}

\begin{tcolorbox}
\label{C-A}
\begin{pcvstack}
 \procedure[  space=auto, mode=pseudocode  ]{\textbf{$\bdv^{\adv}$.\text{queue-init}()} } {
    \text{macs-queue} \leftarrow []\\
    \text{data-queue} \leftarrow []\\
    nonce \sample \bin^\ell \\
    \text{back-index} \leftarrow 0\\
    \text{front-index} \leftarrow 1\\
    \text{enqueued-values} = []
     }

    \procedure[ mode=pseudocode  , space=auto,  valign ]{\textbf{$\bdv^{\adv}$.\text{queue-choose}()} } {
    (x, enqueue) \leftarrow \adv.\text{queue-choose-index-attack}()\\
    \pcif enqueue\\
    \text{data-queue}.\text{enqueue}(x)\\
    \text{back-index} \leftarrow \text{back-index} + 1\\ 
    \text{enqueued-values}.\text{enqueue}(x)\\
    \pcreturn x,(\text{nonce}, \text{back-index}),\\(\text{nonce}, \text{back-index}, \text{front-index})\\
    \pcelse\\
    \text{data-queue}.\text{dequeue}()\\
    \text{mac-queue}.\text{dequeue}()\\
    \text{front-index} \leftarrow \text{front-index} + 1\\ 
    \text{enqueued-values}.\text{dequeue}()\\
    \pcreturn \text{data-queue}.\text{front}(), (\text{nonce}, \text{front-index}), \\\pcind(\text{nonce}, \text{back-index}, \text{front-index})\\
    \pcendif
     }

     \procedure[ mode=pseudocode  ]{\textbf{$\bdv^{\adv}$.\text{queue-receive}(mac)} } {
    \adv.\text{queue-receive}(mac_1,mac_2)
     }

    \procedure[ mode=pseudocode , space=auto,  valign  ]{\textbf{$\bdv^{\adv}$.\text{queue-attack}()} } {
    \pcforeach x' \in \text{enqueued-values}\\
    (x'', mac, k_1, k_2) \leftarrow \adv.\text{queue-index-attack}()\\
    \pcif x' \neq x''\\
    \pcreturn (x'', mac, \text{mac-in-register},(\text{nonce}, k_1),\\ \pcind(\text{nonce}, k_1, k_2))\\
    \pcendif\\
    \pccomment{Updating the state of the authenticated queue.}\\
    \text{front-index} \leftarrow \text{front-index} + 1\\ 
    \pcendfor\\
    \pcreturn (\text{enqueued-values}.\text{front}(),\text{mac-in-register},\\ \pcind\pcind (\text{nonce}, \text{front-index}), \\ \pcind\pcind(\text{nonce}, \text{back-index}, \text{front-index}))
     }
 \end{pcvstack}
 \end{tcolorbox}

\begin{lemma}\label{game3-index-lemma}
\begin{align*}
    &Pr[\textrm{Queue-Game-Index-MAC}^{\bdv^\adv}_2(q) = 1] \\ &\leq Pr[{\textrm{MAC-Collision}^\adv_{\mac_k}}(q) = 1] \le 1 - \frac{2^b!}{(2^b - q)!2^{q.b}}
\end{align*}
\end{lemma}
\begin{proof}
We can reduce the $\textrm{Queue-Game-Index-MAC}^{\bdv^{\adv}}_2$(q) to ${\textrm{MAC-Collision}^\adv_{\mac_k}}(q)$. From lines 13--15 of $\textrm{Queue-Game-Index-MAC}^{\bdv^{\adv}}_2$, winning requires $\bdv^{\adv}$ to find a collision in index \glspl{mac}. The reason is that for the authentication to pass successfully, the \gls{mac} of the proposed indices by \adv\xspace should be equal to the \gls{mac} in the register (in this game, the indices proposed by \adv\xspace are necessarily different from the correct indices since this will allow \adv\xspace to replay previous elements' \glspl{mac}). Moreover, since $\bdv^{\adv}$ winning the game implies that $\adv$ has found a collision (we can replace $\adv$ with $\bdv^{\adv}$ in the collision game and win the game), we can conclude that the bound above holds. 
\end{proof}

$\mathbf{\textbf{\textrm{Queue-Game-Data-MAC}}^\adv_1.}$
This game is similar to $\textrm{Queue-Game-Index-MAC}^\adv_1$, except that, since removing elements from the authenticated queue does not produce new data \glspl{mac}, and there is no need to allow \adv to dequeue elements.
\begin{tcolorbox}
 \begin{gameproof}[name= \text{Queue-Game-Data-MAC}^\adv, arg=(q)]
\begin{pcvstack}[center]
\gameprocedure[ linenumbering, space=auto, mode=pseudocode  ] {
\label{my:line:queuedata1}
\pccomment{We first represent the authenticated queue initialization.}\\ \pccomment{ The data and \gls{mac} queues are unprotected.}\\
\text{macs-queue} \leftarrow []\\
\text{data-queue} \leftarrow []\\
\text{nonce} \sample \text{random}\\
\text{back-index} \leftarrow 0\\
\text{front-index} \leftarrow 1\\
\text{mac-in-register}[] \leftarrow \mac_k(\text{nonce}, \text{back-index},\\ \text{front-index})\\
\pccomment{enqueued-values is a queue that stores the values enqueued into }\\ \pccomment{authenticated queue to be used {\normalfont in} the attack loop.}\\
\text{enqueued-values} = []\\
\pccomment{In the following loop, \adv\xspace experiments with the }\\ \pccomment{queue through the enqueue} \\
\pccomment{operation {\normalfont in} order to collect corresponding \glspl{mac}  {\normalfont for}}\\ \pccomment{ different data values. The \gls{mac} validation steps}\\
\pccomment{have been omitted since \adv's goal is not to attack at }\\ \pccomment{this stage.}\\
\pcfor i \in {1,...,q} \pcdo \\
x \leftarrow \adv.\text{queue-choose-data-attack}()\\
\pccomment{The following steps represent \text{authenticated-queue}.\text{enqueue}(x).}\\
\text{data-queue}.\text{enqueue}(x)\\
\text{back-index} \leftarrow \text{back-index} + 1\\ 
\text{mac-queue}.\text{enqueue}(\mac_k(x,(\text{nonce}, \text{back-index})))\\
\text{mac-in-register} \leftarrow \\ \mac_k(\text{nonce}, \text{back-index}, \text{front-index})\\
\text{enqueued-values}.\text{enqueue}(x)\\
\pccomment{\adv\xspace receives the corresponding \gls{mac} }\\
\adv.\text{queue-receive}(\text{mac-queue}.\text{back}())\\
\pcendfor
}
 \end{pcvstack}
 \end{gameproof}
\end{tcolorbox}
\begin{tcolorbox}
 \begin{gameproof}[name= \text{Queue-Game-Data-MAC}^\adv, arg=(q), nr=0]
\begin{pcvstack}[  center ]
\gameprocedure[lnstart=30, space=auto, linenumbering, mode=pseudocode ,  valign ] {
\pccomment{In the following loop, \adv\xspace attempts to violate the}\\ \pccomment{integrity of the queue by replacing data, and its \gls{mac}.}\\
\pccomment{If the returned data is different from what was originally}\\ \pccomment{ enqueued, and the \gls{mac} verifies, \adv}\\
\pccomment{wins the game, otherwise, they lose. Since the \adv}\\ \pccomment{ is not attacking the indexes {\normalfont in} this}\\ \pccomment{game, we have omitted the verification {\normalfont for} the index mac.}\\
\pcforeach x' \in {\text{enqueued-values}}\\
(x'', mac) \leftarrow \adv.\text{queue-data-attack}() \\
\pcif x' \neq x''\\
\pcif mac = \mac_k(x'',(\text{nonce}, \text{front-index})) \\
\pcreturn 1\\
\pcelse\\
\pcreturn 0\\
\pcendif\\
\pcendif\\
\pccomment{Updating the front-index {\normalfont for} the data \gls{mac} validation }\\ \pccomment{{\normalfont in} the next iteration.}\\
\text{front-index} \leftarrow \text{front-index} + 1\\
\pcendfor\\
\pcreturn 0
}
 \end{pcvstack}
 \end{gameproof}
\end{tcolorbox}

$\mathbf{\textbf{\textrm{Queue-Game-Data-MAC}}^{\bdv^\adv}_2.}$
We now transform the $\textrm{Queue-Game-Data-MAC}^\adv_1$ game into $\textrm{Queue-Game-Data-MAC}^{\bdv^\adv}_2$ by replacing the \gls{mac} function with a random oracle.
\begin{tcolorbox}
 \begin{gameproof}[name= \text{Queue-Game-Data-MAC}^\adv, arg=(q), nr=1]
\begin{pcvstack}[   center ]
\gameprocedure[ linenumbering, space=auto, mode=pseudocode  ] {
\label{my:line:queuedata2}
\text{macs-queue} \leftarrow []\\
\text{data-queue} \leftarrow []\\
\text{nonce} \sample \text{random}\\
\text{back-index} \leftarrow 0\\
\text{front-index} \leftarrow 1\\
\text{mac-in-register} \leftarrow \\ \mac_k(\text{nonce}, \text{back-index}, \text{front-index})\\
\text{enqueued-values} = [] \\
\pcfor i \in {1,...,q} \pcdo \\
x \leftarrow \adv.\text{queue-choose-data-attack}()\\
\text{data-queue}.\text{enqueue}(x)\\
\text{back-index} \leftarrow \text{back-index} + 1\\ 
\pccomment{We replaced the \gls{mac} function from previous games with}\\ \pccomment{a random oracle RO.}\\
\text{mac-queue}.\text{enqueue}(RO(x,(\text{nonce}, \text{back-index})))\\
\text{mac-in-register} \leftarrow  \\ RO(\text{nonce}, \text{back-index}, \text{front-index})\\
\text{enqueued-values}.\text{enqueue}(x)\\
\adv.\text{queue-receive}(\text{mac-queue}.\text{back}())\\
\pcendfor\\
\pcforeach x' \in {\text{enqueued-values}}\\
(x'', mac) \leftarrow \adv.\text{queue-data-attack}()\\
\pcif x' \neq x''\\
\pcif \gls{mac} = RO(x'',(\text{nonce}, \text{front-index})) \\
\pcreturn 1\\
\pcelse\\
\pcreturn 0\\
\pcendif\\
\pcendif\\
\text{front-index} \leftarrow \text{front-index} + 1\\
\pcendfor\\
\pcreturn 0
}
 \end{pcvstack}
 \end{gameproof}
\end{tcolorbox}

\begin{lemma}\label{claim2-queue}
\begin{align*}
    &Pr[\textrm{Queue-Game-Data-MAC}^{\bdv^\adv}_2(q) = 1] \\ &= Pr[\textrm{Queue-Game-Data-MAC}^\adv_1(q) = 1]. 
\end{align*}
\end{lemma}
\begin{proof}
In this game, \adv\xspace attempts to replace an element with another value but with the same index. Accordingly, in order to pass the authentication, \adv\xspace needs to find the corresponding \gls{mac} over the nonce, the new element's value, and the index. As each instance of the data structure has own nonce, and each index only appears once in the lifetime of the data structure, we can conclude that the \gls{mac} required by \adv\xspace has not been previously calculated. Therefore, \adv\xspace is unable to replay a previous \gls{mac} with better than random chance and the probability of \adv\xspace winning both games is the same. 
\end{proof}

\begin{lemma}\label{game3-data-lemma}
\begin{align*}
    Pr[\textrm{Queue-Game-Data-MAC}^{\bdv^\adv}_2(q) = 1] = 2^{-b}.
\end{align*}
\end{lemma}
\begin{proof}
\adv\xspace must find the corresponding output for a query from a random oracle over input values that have not been seen before. Therefore, since $\adv$ does not have access to the random oracle outside the game structure, their only option is to guess the output, which leads to the above probability.
\end{proof}

\begin{theorem}[Queue security] Suppose a program uses the authenticated queue data structure. An adversary with arbitrary read/write control over memory can
violate the integrity of the authenticated queue either by manipulating data, or the index \glspl{mac}, with probability
\begin{align*}
    &Pr[\textrm{Queue-Game-Data-MAC}^\adv_1(q)] \\
                &+ Pr[\textrm{Queue-Game-Index-MAC}^{\bdv^\adv}_1(q)] \\
        \leq \;& 1 - \frac{2^b!}{(2^b - q)!2^{q.b}} + 2^{-b}
\end{align*}
\end{theorem}
\begin{proof}
The probability that an adversary can corrupt \emph{data} read from memory is determined by the game $\textrm{Queue-Game-Data-MAC}^\adv_1$(q); by \Cref{claim2-queue} and \Cref{game3-data-lemma}, \adv's probability of winning $\textrm{Queue-Game-Data-MAC}^\adv_1(q)$  is at most $2^{-b}$.

\Cref{claim1-queue} and \Cref{game3-index-lemma} provide a similar bound on \adv's probability of winning $\textrm{Queue-Game-Index-MAC}^\adv_1(q)$ at \[ 1 - \frac{2^b!}{(2^b - q)!2^{q.b}} . \]

The probability that the attacker can corrupt \emph{either} type of data in memory is therefore at most \[ 1 - \frac{2^b!}{(2^b - q)!2^{q.b}} + 2^{-b}. \]
 \end{proof}
